\documentclass[10pt,reqno]{amsart}
\usepackage{amsmath,amsthm}
\usepackage[colorlinks=true,linkcolor=blue,urlcolor=blue,citecolor=blue]{hyperref}
\usepackage[all]{xy}
\usepackage{amsmath, amssymb, amsfonts, amsthm}
\usepackage{hyperref}
\usepackage[english]{babel}
\usepackage{amsmath}
\usepackage{amssymb}
\usepackage{url}
\usepackage{amsthm}
\usepackage{amsxtra}
\usepackage{mathrsfs}
\usepackage{fancyhdr}
\usepackage{enumerate}
\usepackage{graphicx}
\usepackage{hyperref}
\usepackage{tikz}
\usepackage{bm}
\newtheorem{theorem}{Theorem}[section]

\newtheorem{proposition}[theorem]{Proposition}
\newtheorem{corollary}[theorem]{Corollary}

\theoremstyle{definition}
\newtheorem{remark}[theorem]{Remark}

\newtheorem{assumption}{Assumption}

\numberwithin{equation}{section}

\newcommand{\form}{\mathcal{F}_{\alpha}}
\newcommand{\formb}{\mathcal{F}_{\alpha,\beta}}
\newcommand{\formf}{\mathcal{F}_{\alpha,\mathrm{F}}}
\newcommand{\formt}{\widetilde{\mathcal{F}}_{\alpha,\beta}}

\newcommand{\formv}{\mathcal{F}_{\alpha,V}}
\newcommand{\formvb}{\mathcal{F}_{\alpha,\beta,V}}
\newcommand{\formvf}{\mathcal{F}_{\alpha,\mathrm{F},V}}

\newcommand{\ham}{H_{\alpha}}
\newcommand{\Ham}{\mathcal{H}_{\alpha}}
\newcommand{\hamv}{H_{\alpha,V}}
\newcommand{\hamb}{H_{\alpha,\beta}}
\newcommand{\hamvb}{H_{\alpha,\beta,V}}
\newcommand{\hamf}{H_{\alpha,\mathrm{F}}}
\newcommand{\hamvf}{H_{\alpha,\mathrm{F}, V}}
\newcommand{\hamg}{H_{\alpha, \gamma}}

\newcommand{\phila}{\phi_{\la}}
\newcommand{\gila}{G_{\la}}
\newcommand{\psiap}{\psi_{\alpha,+}}
\newcommand{\psiam}{\psi_{\alpha,-}}

\newcommand{\cal}{c_{\alpha}}

\newcommand{\Le}{L^2_{\mathrm{even}}}
\newcommand{\nablaa}{\mathbf{D}_{\alpha}}

\newcommand{\xcm}{X_{\mathrm{cm}}}

\newcommand{\beq}{\begin{equation}}
\newcommand{\eeq}{\end{equation}}
\newcommand{\be}{\begin{equation*}}
\newcommand{\ee}{\end{equation*}}

\newcommand{\OO}{\mathcal{O}}

\newcommand{\FF}{\mathcal{F}}
\newcommand{\Z}{\mathbb{Z}}

\newcommand{\lf}{\left}
\newcommand{\ri}{\right}

\newcommand{\la}{\lambda}

\renewcommand{\Re}{\operatorname{Re}\,}

\renewcommand{\leq}{\leqslant}
\renewcommand{\geq}{\geqslant}

\newcommand{\HH}{\mathscr H}
\newcommand{\C}{\mathbb{C}}

\newcommand{\xv}{\mathbf{x}}

\newcommand{\rv}{\mathbf{r}}

\newcommand{\bdm}{\begin{displaymath}}
\newcommand{\edm}{\end{displaymath}}
\newcommand{\bdn}{\begin{eqnarray}}
\newcommand{\edn}{\end{eqnarray}}
\newcommand{\bay}{\begin{array}{c}}
\newcommand{\eay}{\end{array}}
\newcommand{\ben}{\begin{enumerate}}
\newcommand{\een}{\end{enumerate}}
\newcommand{\beqn}{\begin{eqnarray}}
\newcommand{\eeqn}{\end{eqnarray}}
\newcommand{\bml}[1]{\begin{multline} #1 \end{multline}}
\newcommand{\bmln}[1]{\begin{multline*} #1 \end{multline*}}

\newcommand{\R}{\mathbb{R}}
\newcommand{\N}{\mathbb{N}}

\newcommand{\diff}{\mathrm{d}}
\newcommand{\dom}{\mathscr{D}}
\newcommand{\eps}{\varepsilon}

\newcommand{\disp}{\displaystyle}
\newcommand{\tx}{\textstyle}

\newcommand{\braket}[2]{\lf\langle #1|#2 \ri\rangle}
\newcommand{\braketr}[2]{\lf\langle #1\lf|#2\ri. \ri\rangle}

\newcommand{\meanlrlr}[3]{\lf\langle #1\lf|#2\ri|#3\ri\rangle}

\begin{document}

\title{Hamiltonians for Two-Anyon Systems}

\author{M. Correggi}
\address{Dipartimento di Matematica, ``Sapienza'' Universit\`a di Roma, P.le A. Moro 5, 00185 Roma, Italy}
\email{michele.correggi@gmail.com}
\urladdr{http://www1.mat.uniroma1.it/people/correggi/}

\author{L. Oddis}
\address{Dipartimento di Matematica, ``Sapienza'' Universit\`a di Roma, P.le A. Moro 5, 00185 Roma, Italy}
\email{oddis@mat.uniroma1.it}

\date{\today}

%\pacs{03.75.Ss, 05.30.Fk, 67.85.-d }


\begin{abstract}
We study the well-posedness of the Hamiltonian of a system of two anyons in the magnetic gauge. We identify all the possible quadratic forms realizing such an operator for non-interacting anyons and prove their closedness and boundedness from below. We then show that the corresponding self-adjoint operators give rise to a one-parameter family of extensions of the naive two-anyon Schr\"{o}dinger operator. We finally extend the results in presence of a two-body radial interaction. 
\end{abstract}	

\dedicatory{Dedicated to Gianfausto Dell'Antonio on the occasion of his 85th birthday}

\maketitle

\section{Introduction}
\label{sec:intro}

The possible existence in quantum mechanics of two-dimensional identical particles obeying to {\it fractional} or {\it intermediate statistics} (later named {\it anyons} \cite{Wi}) was known since the pioneering work \cite{LM} of J.M. Leinaas and J. Myrheim. Whether such particles could exist in nature or play a role in physical models remained however an open question, until they were suggested as quasi-particle carriers in a model for the fractional quantum Hall effect \cite{ASW}. Since then several physical theories involving anyons have been proposed to describe phenomena of condensed matter physics (see, e.g., \cite{LunRou2} for an example). 

Meanwhile, the mathematical implications of fractional statistics have been studied as well (see, e.g., \cite{BCMS, CLR, GHKL, LM, LarLun, Lu, LunRou1, LS, LSo, MT}), but despite the numerous results on the topic, the question of rigorous definition of the Hamiltonian for a many-anyon system has not been studied in detail. This is precisely the problem we deal with in this note for the simplest system of two anyons.

Let us now discuss in more details the question we plan to study: in the {\it magnetic gauge} the state of a two-anyon system is described by a wave function $ \Psi(\xv_1, \xv_2) \in L^2_{\mathrm{sym}}(\R^4) $, i.e., a square integrable function which is symmetric under exchange $ \xv_1 \to \xv_2 $. The Hamiltonian of the system $ \Ham $ acts on $ L^2_{\mathrm{sym}}(\R^4) $ as
\beq
	\label{eq: Ham}
	\Ham = \lf( - i \nabla_1 + \tx\frac{\alpha (\xv_1 - \xv_2)^{\perp}}{|\xv_1 - \xv_2|^2} \ri)^2 + \lf( - i \nabla_2 + \tx\frac{\alpha (\xv_2 - \xv_1)^{\perp}}{|\xv_1 - \xv_2|^2} \ri)^2 + V\lf(\lf|\xv_1 - \xv_2\ri|\ri),
\eeq
where $ \alpha \in [0,1] $ is the {\it statistic parameter} (with $ \alpha = 0,1 $ identifying {\it bosons} and {\it fermions}, respectively), $ V $ is the interaction potential, $ \xv^{\perp} : = (-y,x) $ and we have set $ m = 1/2 $ and $ \hbar = 1 $. Hence, each particle generates a {\it Aharonov-Bohm} (AB) {\it magnetic potential}, affecting the other one. The intensity of such a field is proportional to the statistic parameter $ \alpha $.

 After the extraction of the center of mass, i.e., setting $ \xcm : = \frac{1}{2}(\xv_1 + \xv_2) $ and $ \rv : = \xv_1 - \xv_2 $, the space of states becomes $ L^2(\R^2) \otimes \Le(\R^2) $, where the latter one denotes the Hilbert space of square-integrable {\it even} functions. Note that the bosonic symmetry constraint translates into the parity request. The operator $ \Ham $ becomes then $ \Ham = 2 ( - \frac{1}{4} \Delta_{\xcm} + \hamv(\rv)) $, where
\beq
	\label{eq: ham}
	\hamv : = \lf( - i \nabla_\rv + \frac{\alpha \rv^{\perp}}{r^2} \ri)^2 + V(r),
\eeq
acts only on $ \Le(\R^2) $.

In the rest of the paper, we study the operator $ \hamv $ and, in particular, focus on its self-adjoint extensions. It is indeed easy to see that, at least if $ V = 0 $ (we set $ \ham : = H_{\alpha,0} $), $ \ham $ is symmetric and positive. Therefore, it certainly admits a self-adjoint extension, i.e., the {\it Friedrichs extension}, which is typically the one selected in almost all mathematical investigations of anyons. In fact, it is a bit harder to realize that $ \ham $ is actually not essentially self-adjoint, e.g., on $ C^{\infty}_0(\R^2 \setminus \{0 \}) $, which is a natural dense domain. The reason for this lack of self-adjointness is obviously the singularity at $ r = 0 $ of the AB potential. In order to properly set up a two-anyon model (e.g., select the dynamics), one has thus to know more about the possible self-adjoint realizations of $ \ham $.

This question has already been investigated in \cite{AT} in the framework of Von Neumann operator theory, although the physical model considered there is slightly different (see below), and in a more heuristic way in \cite{BS}. Here, we take a different point of view and introduce a one-parameter family of {\it quadratic forms}, which are meant to describe the possible realizations of the energy of the system. Next, we prove in Thm. \ref{teo: closed} that such forms are closed and bounded from below. Finally, in Cor. \ref{cor: hamb} we show that the corresponding operators are self-adjoint extensions of $ \ham $ and, in fact, exhaust all such extensions. The results are then extended to the interacting case, under suitable assumptions on  $ V $.

%\medskip

We briefly recall here some notation that will be used in the rest of the paper. Given two functions $ f(x), g(x) $, with $ g > 0 $, we use the following convention for Landau symbols: 
%\begin{itemize}
	%\item 
	$ f = \OO(g) $, if $ \lim_{x \to 0^+} |f|/g \leq C $;
	%\item 
	$ f = o(g) $, if $  \lim_{x \to 0^+} |f|/g = 0 $; 
	%\item 
	$ f \sim g $, if $ f = \OO(g) $ and $ \lim_{x \to 0^+} |f|/g > 0 $.
%\end{itemize}
Here and below $ C $ stands for a finite positive constant, whose value may change from line to line.

\section{Main Results}
\label{sec: main}

%In this Sect. we formulate our main results. We first discuss the simple but key case of non-interacting anyons (Sect. \ref{sec: free}) and then generalize the results to interacting particles (Sect. \ref{sec: interacting}).

\subsection{Free anyons}
\label{sec: free}

Before providing the definition of the quadratic forms describing the possible realizations of the center of mass energy of a pair of anyons, we first have to introduce a quadratic form $ \formf $ which is associated to a very special extension of $ H_{\alpha} $, i.e., the Friedrichs extension $ \hamf $: $ \ham $ is a positive symmetric operator and, as such, it admits at least one self-adjoint extension which is still positive and whose domain is contained in the domain of the corresponding quadratic form, simply defined as the expectation value of $ \ham $, i.e.,
\beq
	\label{eq: formf}
	\formf[\psi] : = \form[\psi] = \int_{\R^2} \diff \rv \: \bigg| \bigg( -i \nabla + \frac{\alpha \rv^{\perp}}{r^2} \bigg) \psi \bigg|^2,
\eeq
with domain
\beq
	\label{eq: dom formf}
	\dom[\formf] = \overline{C^{\infty}_0(\R^2 \setminus \{ 0 \})}^{\lf\| \; \ri\|_{\alpha}} \cap \Le. %\lf\{ \psi \in L^2(\R^2) \: \big| \: \form[\psi] < + \infty \ri\}.
\eeq
Here, $ \lf\| \phi \ri\|_{\alpha}^2 : = \form[\phi] $ and 
\beq
	\label{eq: hilb}
	\Le : = \lf\{ \psi \in L^2(\R^2) \: \big| \: \psi(-\rv) = \psi(\rv) \ri\} = \bigoplus_{k \in \Z} \HH_{2k},
\eeq
where we have denoted for short
\beq
	\HH_{n} : = L^2(\R^+, r \diff r) \otimes \mathrm{span} \lf( e^{i n \vartheta} \ri),
\eeq
and used polar coordinates $ \rv = (r, \vartheta) \in \R^+ \times [0,2\pi) $. Throughout the paper, we will refer to the decomposition in \eqref{eq: hilb} by setting
\beq
	\label{eq: fourier}
	\psi(\rv) = \frac{1}{\sqrt{2\pi}} \sum_{k \in \Z} \psi_{2k}(r) e^{2 i k \vartheta},
\eeq
and use the notation $ \psi_n $ to denote the Fourier coefficients of $ \psi $. With respect to such a decomposition, the quadratic form \eqref{eq: formf} can be rewritten as
\beq
	\label{eq: formf alternative}
	\formf[\psi] = \form[\psi] = \sum_{k \in \Z} \int_{0}^{+\infty} \diff r  \: r \, \bigg\{ \lf| \psi_{2k}^{\prime} \ri|^2 + \frac{\lf( 2 k  + \alpha \ri)^2}{r^2} \lf| \psi_{2k} \ri|^2 \bigg\}.
\eeq

As a preliminary result, we prove that $ \formf $ is closed on $ \dom[\formf] $ and characterize its domain. Notice that at this stage the name Friedrichs extension is not justified and it will make sense only once the whole family of forms is defined.

	\begin{proposition}[Friedrichs extension]
		\label{pro: formf}
		\mbox{}	\\
		The quadratic form $ \formf $ is closed and positive on $ \dom[\formf] $ for any $ \alpha \in [0,1] $. Furthermore, for any $ \alpha \in (0,1) $,
		\beq
			\label{eq: dom formf zero}
			\dom[\formf] \subset %H_0^1(\R^2 \setminus \{ 0 \})=
			H^1(\R^2).
		\eeq
		The associated self-adjoint operator $ \hamf $ acts as $ \ham $ on the domain
		\beqn
			\label{eq: dom hamf}
			\dom\lf( \hamf \ri) &=& \lf\{ \psi \in \dom\lf[ \formf \ri] \: \big| \: \ham \psi \in L^2 \ri\} = \lf\{ \psi \: \Big| \: \lf. \psi \ri|_{\HH_n} \in H^2(\R^2), \forall n \neq 0; \ri. \nonumber	\\
			&&  \lf. \psi_0 \in H^2(\R^2\setminus\{0\}) \cap H^1(\R^2), \, \psi_0(r)  \underset{r \to 0^+}{\sim} r^{\alpha} + o(r) \ri\}.
		\eeqn
	\end{proposition}
	
	\begin{remark}[Asymptotics for $ r \to 0 $]
		\label{rem: dom hamf asympt r}
		\mbox{}	\\
		Notice that the asymptotics in \eqref{eq: dom hamf} applies only to the $ s-$wave component of $ \psi $, i.e., $ \psi_0 $. In fact, the other Fourier components $ \psi_n $, $ n \neq 0 $, are such that 
		\beq
			\psi_n(r) = o(r),		\qquad \mbox{as } r \to 0^+,
		\eeq
		i.e., functions in $ \dom\lf( \hamf \ri) $ with non-zero angular momentum must vanish at $ 0 $ faster than $ r $. 
	\end{remark}

	As we are going to see the form $ \formf $ will be contained in the family of forms we are about to define and, among such forms, $ \formf $ is the largest, namely the one with the largest lower bound, but also the one with smallest domain. 
	
	The family of quadratic forms $ \formb[\psi] $, $ \alpha \in (0,1) $ and $ \beta \in \R $, is defined as
\beq
	\label{eq: formb}
		\formb[\psi] : = \form[\phila] - 2 \la^2 \Re q \braket{\phila}{\gila} + \lf[ \beta + \lf( 1 - \alpha \ri) \cal \la^{2\alpha} \ri] \lf| q \ri|^2
\eeq
where $ \psi $ belongs to the domain
\beq
	\label{eq: dom formb}
	\dom[\formb] = \lf\{ \psi \in \Le \: \big| \: \psi = \phila + q \gila, \phila \in \dom[\formf], q \in \C \ri\},
\eeq
and $ G_{\la} $ , $ \la \in \R^+ $, is the defect function
\beq
	\label{eq: gila}
	\gila(\rv) : = \la^{\alpha} K_{\alpha}(\la r),
\eeq
with $ K_{\alpha} $ the modified Bessel function of index $ \alpha $ (see, e.g., \cite[Sect. 9.6]{AS}). The coefficient $ \cal $ is given by (see \cite[Eq. 6.521.3]{GR})
\beq
	\label{eq: cal}
	\cal : = \frac{\lambda^{2-2\alpha} \lf\| \gila \ri\|_2^2}{\alpha} = \frac{2\pi}{\alpha} \int_0^{+\infty} \diff r \: r \: \lf| K_{\alpha}(r) \ri|^2 = \frac{\pi^2}{\sin \pi \alpha} > 0
\eeq
and chosen in such a way that the form is independent of $ \la $ (see below). The key property of $ \gila $ is that
\beq
	\label{eq: defect}
	\lf( - \nablaa^2 + \la^2 \ri) \gila = 0,	
\eeq
where we have used the short-hand notation
\beq
	\label{eq: nablaa}
	\nablaa : = - i \nabla + \frac{\alpha \rv^{\perp}}{r^2}.
\eeq
The identity \eqref{eq: defect} is indeed a consequence of the Bessel equation \cite[Eq. 9.6.1]{AS}
\bdm
	z^2 \partial_z^2 K_{\alpha}(z) + z K_{\alpha}(z) - (z^2 + \alpha^2) K_{\alpha}(z) = 0.
\edm		
Notice that (see \cite[Eqs. 9.6.9]{AS} or, more precisely, \cite[Eqs. 8.440 \& 8.443]{GR} and the definition of $ K_{\alpha} $)
\beq
	\label{eq: gila asympt}
	\gila(\rv) = 2^{\alpha - 1} \Gamma(\alpha) r^{-\alpha} - \frac{\Gamma(1-\alpha) \lambda^{2\alpha}}{ 2^{1+\alpha} \alpha} r^{\alpha} + \OO(r^{2-\alpha}), 		\qquad		\mbox{as } r \to 0^+,
\eeq
(in particular $ G_{\la}\not\in H^1(\R^2)$) while, for large $ r $ and $ \la $ positive, $ \gila(\rv) \sim r^{-1/2} e^{-r} $ \cite[Eq. 9.7.2]{AS}, so that
\beq
	\label{eq: gila l2}
	\gila \in L^2(\R^2),		\qquad		\mbox{if and only if }	\alpha \in [0,1).
\eeq

For $ \alpha \in (0,1) $, the quadratic form \eqref{eq: formb} can in fact be rewritten as
\beq
	\label{eq: formb alternative}
	\framebox{$ \formb[\psi] = \form[\phila] +\la^2 \lf\| \phila \ri\|^2_2 - \la^2 \lf\| \psi \ri\|_2^2 + \lf( \beta + \cal \la^{2\alpha} \ri) \lf| q \ri|^2, $}
\eeq
As anticipated, the Friedrichs form $ \formf $ is included in the family and formally recovered for $ \beta = + \infty $, in which case $ q = 0 $ and $ \dom[\FF_{\alpha, + \infty}] = \dom[\formf] $.

Before discussing the properties of the quadratic form $ \formb $, we have in fact to show that the definition \eqref{eq: formb} is well-posed. First of all the decomposition $ \phila + q \gila $ is unique, since $ \phila \in \dom[\formf] $, while $ \gila \in L^2(\R^2) \setminus \dom[\formf] $ (see \eqref{eq: gila asympt} above). Furthermore, in spite of the presence of the parameter $ \la \in \R^+ $, the form \eqref{eq: formb} along with its domain \eqref{eq: dom formb} is in fact independent of $ \la $: let $ \la_1 \neq \la_2 \in \R^+ $, then
\beq
	G_{\la_1}(\rv) - G_{\la_2}(\rv) \in \dom\lf[ \formf \ri],
\eeq
which is again a consequence of \eqref{eq: gila asympt}, since $ G_{\la_1}(\rv) - G_{\la_2}(\rv) = \OO(r^{\alpha}) $, as $ r \to 0^+ $. 
Now, let $\psi=\phi_{\la_1} + q G_{\la_1}=\phi_{\la_2}+qG_{\la_2}$, then we get
\bmln{
	%\label{eq: independent}
	\formb[\psi]=\form \lf[\phi_{\la_2}+q\lf(G_{\la_2}-G_{\la_1}\ri) \ri] + \la_1^2 \lf\| \phi_{\la_1} \ri\|^2_2 - \la_1^2 \lf\| \psi \ri\|^2_2 + \lf( \beta + \cal \la_1^{2\alpha} \ri) \lf| q \ri|^2 \\
	%= \form\lf[ \phi_{\la_2} \ri] - 2 \Re \lf( \form\lf[\phi_{\la_2}\:,\:q\lf(G_{\la_2}-G_{\la_1}\ri) \ri] \ri) - \form \lf[ q\lf(G_{\la_2}-G_{\la_1}\ri) \ri] \\ 
	%+ \la_1^2 \lf\| \phi_{\la_1} \ri\|^2_2 - \la_1^2 \lf\| \psi \ri\|^2_2 + \lf( \beta + \cal \la_1^{2\alpha} \ri) \lf| q \ri|^2 \\
	= \form\lf[ \phi_{\la_2} \ri] - 2 \Re q \braketr{\phi_{\la_2}}{\la_2^2 G_{\la_2}- \la_1^2 G_{\la_1}} - \lf| q \ri|^2 \braket{G_{\la_2}-G_{\la_1}}{\la_2^2G_{\la_2}- \la_1^2G_{\la_1}} \\ 
	+ \la_1^2 \lf\| \phi_{\la_2} +q\lf(G_{\la_2}-G_{\la_1}\ri) \ri\|^2_2 - \la_1^2 \lf\| \psi \ri\|^2_2 + \lf( \beta + \cal \la_1^{2\alpha} \ri) \lf| q \ri|^2, }
	where we used \eqref{eq: defect}. Writing now 
	\bdm
		2 \Re q \braketr{\phi_{\la_2}}{ G_{\la_2}} = \lf\| \psi \ri\|^2 - \lf\| \phi_{\la_2} \ri\|^2_2 - \lf| q \ri|^2 \lf\|  G_{\la_2} \ri\|_2^2
	\edm
	and exploiting the fact that $ \gila $ is real, we obtain
\bmln{
	%\label{eq: independent}
	%= \form\lf[ \phi_{\la_2} \ri] + \la_2^2 \lf\| \phi_{\la_2} \ri\|^2_2  - \la_2^2 \lf\| \psi \ri\|^2_2 + \la_1^2 \lf[ 2 \Re q \braketr{\phi_{\la_2}}{G_{\la_1}} +  \lf\| \phi_{\la_1} \ri\|^2_2 - \lf\| \psi \ri\|_2^2 \ri] \\ 
	 %+ \lf[ \beta + \cal \la_1^{2\alpha} + \lf( \la_1^2 + \la_2^2 \ri) \braket{G_{\la_1}}{G_{\la_2}} - \la_1^2 \lf\| G_{\la_1} \ri\|_2^2 \ri] \lf| q \ri|^2 \\
	 %= \form\lf[ \phi_{\la_2} \ri] + \la_2^2 \lf\| \phi_{\la_2} \ri\|^2_2  - \la_2^2 \lf\| \psi \ri\|^2_2 + \la_1^2 \lf[ 2 \Re q \braketr{\phi_{\la_2}-\phi_{\la_1}}{G_{\la_1}} \ri] \\ 
	 %+ \lf[ \beta + \cal \la_1^{2\alpha} + \lf( \la_1^2 + \la_2^2 \ri) \braket{G_{\la_1}}{G_{\la_2}} - 2 \la_1^2 \lf\| G_{\la_1} \ri\|_2^2 \ri] \lf| q \ri|^2 \\
	 \formb[\psi] = \form\lf[ \phi_{\la_2} \ri] + \la_2^2 \lf\| \phi_{\la_2} \ri\|^2_2  - \la_2^2 \lf\| \psi \ri\|^2_2 + 2 |q|^2 \braketr{\la_1^2 G_{\la_1} - \la_2^2 G_{\la_2}}{G_{\la_1}}  \\ 
	 + \lf[ \beta + \cal \la_1^{2\alpha} + \lf( \la_1^2 + \la_2^2 \ri) \braket{G_{\la_1}}{G_{\la_2}} - 2 \la_1^2 \lf\| G_{\la_1} \ri\|_2^2 \ri] \lf| q \ri|^2 \\
	  = \form\lf[ \phi_{\la_2} \ri] + \la_2^2 \lf\| \phi_{\la_2} \ri\|^2_2  - \la_2^2 \lf\| \psi \ri\|^2_2  + \lf[ \beta + \cal \la_1^{2\alpha} + \lf( \la_1^2 - \la_2^2 \ri) \braket{G_{\la_1}}{G_{\la_2}} \ri] \lf| q \ri|^2,
}
which reproduces the expression of the quadratic form w.r.t. the $ \la_2$-decomposition, via the identity \cite[Eq. 6.521.3]{GR}
\bdm
	 \lf( \la_1^2 - \la_2^2 \ri)\braket{G_{\la_1}}{G_{\la_2}} = \cal \lf( \la_2^{2\alpha} - \la_1^{2\alpha} \ri).
\edm

\begin{remark}[Fermions]
	\label{rem: fermions}
	\mbox{}	\\
	Throughout the paper, we will always assume that $ \alpha \in [0,1) $, i.e., we exclude the case of fermions $ \alpha = 1 $, since in this case $ H_{\alpha} $ is essentially self-adjoint on $ C^{\infty}_0(\R^2 \setminus \{0 \}) \cap \Le $ and therefore no other extension is admissible. From the point of view of quadratic forms, this is made apparent from the fact that \eqref{eq: formb} and \eqref{eq: dom formb} are ill-defined for $ \alpha =  1$, since  $ \gila \notin  L^2(\R^2) $.
\end{remark}

%Our main result is the following

	\begin{theorem}[Closedness and boudedness from below of $ \formb $]
		\label{teo: closed}
		\mbox{}	\\
		For any $ \alpha \in (0,1) $ and any $ \beta \in \R $, the quadratic forms $ \formb $ is closed and bounded from below on the domain $ \dom[\formb] $. Furthermore,
		\beq
			\label{eq: lb formb}
			\frac{\formb[\psi]}{\lf\| \psi \ri\|_2^2} \geq
			\begin{cases}
				0,		&		\mbox{if } \beta \geq 0;	\\
				- \lf( \disp\frac{|\beta| \sin(\pi\alpha)}{\pi^2} \ri)^{\frac{1}{\alpha}},		&		\mbox{if }	 \beta < 0.
			\end{cases}
		\eeq
	\end{theorem}
	
	\begin{remark}[Bosons]
		\label{rem: bosons}
		\mbox{}	\\
		The above Theorem does not cover the case of bosons $ \alpha = 0 $ simply because it can not be obtained as the (formal) limit $ \alpha \to 0^+ $. This is indeed made apparent from the fact that $ \cal \to  + \infty $, as $ \alpha \to 0^+ $, which generates a divergence in the last term of the quadratic form. This is in turn related to the change of behavior at the origin of the defect function $ \gila $, when $ \la^{\alpha} K_{\alpha}(\la r) $ is replaced with $ K_0(\la r) $. The asymptotics of the latter as $ r \to 0^+ $ is indeed $ K_0(\la r) = - \log r + \log \la/2 + \gamma_{\mathrm{E}} + \OO(r^2) $, $ \gamma_{\mathrm{E}} $ being the Euler constant.
		
		We do expect however that, if the quadratic form is properly modified to take into account the different properties of the defect function $  K_0(\la r) $, its closedness and boundedness from below hold true as well and the energy form (more precisely its restriction to $ L^2_{\mathrm{even}}(\R^2) $) of a particle with a point interaction at the origin (see, e.g., \cite{CCT}) is recovered.
	\end{remark}
	
	As we are  going to see the bound from below stated in \eqref{eq: lb formb} is sharp, i.e., it also provides the ground state energy (in fact, the energy of the unique negative eigenvalue) of the corresponding self-adjoint operator.
	
	Indeed, being $ \formb $ a family of closed and bounded from below quadratic forms, they uniquely identify a family of self-adjoint operators which are bounded from below as well:
	
	\begin{corollary}[Self-adjoint operators $ \hamb $]
		\label{cor: hamb}
		\mbox{}	\\
		The one-parameter family of self-adjoint operators associated to the forms $ \formb $, $ \alpha \in (0,1) $ and $ \beta \in \R $, is given by
		\beq
			\label{eq: hamb action}
			 \lf( \hamb + \la^2 \ri) \psi = \lf( \ham + \la^2 \ri) \phila,	
		\eeq
		\bml{
			  \dom\lf( \hamb \ri) =  \lf\{ \psi \in \Le \: \Big| \: \psi_n \in \dom\lf(\hamf\ri), \forall n \neq 0; \; \psi_0 = \phila + q \gila,  \ri. 	\\
			  \lf. \phila \in \dom\lf(\hamf\ri), \; q = - \frac{\Gamma(\alpha)}{2^{1-\alpha}(\beta + \cal \la^{2\alpha})}  \lim_{r \to 0^+} \frac{\alpha \phila(r) + r \phila^{\prime}(r)}{r^{\alpha}} \ri\},	 \label{eq: dom hamb}
		}
		where $ \lambda > 0 $ is free to choose provided $ \beta + \cal \la^{2\alpha} \neq 0 $.
		\newline
		Furthermore, the operators $ \hamb $ extend $ \ham $, i.e., $ \lf. \ham \ri|_{\dom(\ham)} = \ham $, and, conversely, any self-adjoint extension of $ \ham $ is included in the family $ \hamb $, $ \beta \in \R $.
	\end{corollary}

	\begin{remark}[Boundary condition]
		\label{rem: bc}
		\mbox{}	\\
		Functions in the domain of $ \hamb $ have thus to satisfy a suitable boundary condition at $ 0 $. Note that such a condition is well-posed since $ \phi_{\la} $ is radial (it belongs to the zero-angular momentum sector) and in $ H^2(\R^2) $ (see Proposition \ref{pro: formf}). In fact, we can assume that 
		%\beq
			$\phila(r) \underset{r \to 0^+}{=} d_{\la} r^{\alpha} (1 + o(1)) $,
		%\eeq
		for some $ d_{\la} \in \R $. Then, %the boundary condition in 
		\eqref{eq: dom hamb} reads
		\beq
			q = - \frac{2^{\alpha} \alpha \Gamma(\alpha)}{\beta + \cal \la^{2\alpha}} d_{\la},
		\eeq
		so making apparent that $ q = 0 $, if $ \phila $ vanishes faster than $ r^{\alpha} $.
	\end{remark}
	
	\begin{remark}[$s-$wave perturbation]
		\label{rem: s wave}
		\mbox{}	\\
		It is clear from the domain definition \eqref{eq: dom hamb} that $ \hamb $ defines an $s-$wave perturbation of the operator $ \ham $: it is indeed only the $0-$momentum Fourier component $ \psi_0 $ which has to satisfy the decomposition as in \eqref{eq: dom hamb}. Furthermore, $ \lf. \hamb \ri|_{\HH_{2k}} = \hamf $, for any $ k \in \Z, k \neq 0 $.
	\end{remark}
	
	Since all the quadratic forms of the positive extensions (i.e., those with $ \beta \geq 0 $) have the same domain, it is not immediate to identify the Krein extension \cite[Thm. 13.12]{S}. However, in perfect analogy with point interactions,	it seems natural to associate the Krein extension to the one labelled by $\beta=0$, since it marks the threshold for the existence of negative bound states. Furthermore, from \eqref{eq: dom hamb}, one can see that both the limits $\beta\to+\infty$ and $\beta\to-\infty$ recover the Friedrichs extension, since in both cases the charge $q$ is identically zero.
	%\begin{remark}[Bosons]
		%\label{rem: bosons}
		%\mbox{}	\\
		%The case of bosons $ \alpha = 0 $ is excluded from the statement of Corollary \ref{cor: hamb} because the boundary condition in \eqref{eq: dom hamb} has to be suitably modified. In fact, the proof of Corollary \ref{cor: hamb} applies to $ \alpha = 0 $ too (see Remark \ref{rem: bosons proof}) and the domain of a two-dimensional point interaction at the origin is recovered\footnote{Actually, the boundary condition as written, e.g., in \cite{CCT}, is slightly different from \eqref{eq: bc bosons}, because of the additive factor, but the more familiar one can be easily recovered via a redefinition of the parameter $ \beta $.} of , i.e.,
		%\beq
			%\label{eq: bc bosons}
			%\dom\lf( H_{0,\beta} \ri) = \lf\{ \psi \: \big| \: \psi = \phila + q \gila, \phila \in \dom\lf(\hamf\ri), \phila(0)  = \lf( \beta + \pi \ri) q \ri\}.	
		%\eeq
		%Combined with the action described in \eqref{eq: dom hamb}, this indeed shows that $ H_{0,\beta} $ is a self-adjoint realization of the formal operator $ - \Delta + a \delta(\rv) $.
	%\end{remark}		

	The second part of the Corollary, i.e., the fact that the family $ \hamb $ exhausts all possible self-adjoint realizations of $ \ham $ is in fact a consequence of the results proven in \cite{AT}, where the self-adjoint extensions of $ \ham $ in $ L^2(\R^2) $ are thoroughly investigated. Notice that the physical goal in \cite{AT} was not to study anyons but rather rigorously derive the Hamiltonian of a quantum particle in a Aharonov-Bohm (AB) magnetic flux centered at the origin. Formally, the operator coincides with $ \ham $ but the space of states is much wider, since no symmetry restriction is imposed. In \cite{AT} it is proven that $ \ham $ admits in $ L^2(\R^2) $ a four-parameter family of self-adjoint extensions \cite[Eqs. (2.18) and (2.19)]{AT}, whose domains and actions can be explicitly characterized: all extensions are indeed perturbations of the operator $ \hamf $ living in the $s-$ and $p-$wave subspaces. This can be seen by studying the deficiency spaces associated to $ \ham $, which are spanned by four functions belonging to $ \HH_0 $ and $ \HH_{-1} $, respectively. 
	
	Once the operator $ \ham $ is restricted to $ \Le $ and functions with odd angular momentum forbidden, only two solutions of the deficiency equations survive and, consequently, a one-parameter family of self-adjoint extensions is left: let us denote by $ \gamma \in [0, 2\pi) $ a real parameter, then the operator family is given by
\beq
	\hamg \psi =  \ham \phi + i \mu \psiap - i \mu e^{i\gamma} \psiam,	
\eeq
\beq
	\label{eq: dom hamg}
	\dom\lf( \hamg \ri) =  \lf\{ \psi \: \Big| \: \psi = \phi + \mu \psiap + \mu e^{i\gamma} \psiam, \phi \in \dom\lf(\hamf\ri), \mu \in \C \ri\},	 
\eeq
where the ($L^2$-normalized) deficiency functions $ \psi_{\alpha,\pm} $ are
\beqn
	\psi_{\alpha,+}(r) &=& \tx\frac{\sqrt{2 \cos(\pi\alpha/2)}}{\pi} \, K_{\alpha}\big( e^{-i\pi/4} r\big),	\nonumber	\\
	\psi_{\alpha,-}(r) &=& \tx\frac{\sqrt{2 \cos(\pi\alpha/2)}}{\pi} \, e^{i \pi \alpha/2}  \, K_{\alpha}\big( e^{i\pi/4} r\big).
\eeqn
As expected, $ \hamg $ differs from $\ham $ only in the subspace with zero angular momentum, since $ \psi_{\alpha,\pm} \in \HH_0 $. In the proof of Cor. \ref{cor: hamb} we show that the family \eqref{eq: dom hamg} is in fact contained in \eqref{eq: dom hamb} and thus the two must coincide.

	\begin{proposition}[Spectral properties of $ \hamb $]
		\label{pro: spectrum}
		\mbox{}	\\
		For any $ \alpha \in (0,1) $ and any $ \beta \in \R $, $ \sigma \lf( \hamb \ri) = \sigma_{\mathrm{pp}}\lf( \hamb \ri) \cup \sigma_{\mathrm{ac}}\lf( \hamb \ri) $, with $ \sigma_{\mathrm{ac}}\lf( \hamb \ri) = \R^+ $ and 
		\beq
			\label{eq: eigenvalue}
			\sigma_{\mathrm{pp}}\lf(\hamb\ri) =
			\begin{cases}
				\lf\{ - \lf( \disp\frac{|\beta| \sin(\pi\alpha)}{\pi^2} \ri)^{\frac{1}{\alpha}} \ri\},		&	\mbox{if } \beta < 0;	\\
				\emptyset,			&	\mbox{otherwise}.
			\end{cases}
		\eeq
		%Furthermore, the resolvent $ \lf( \hamb + \la^2 \ri)^{-1} $ of $ \hamb $ for $ \la \in \R^+ $ is given by the integral operator with kernel
		%\beq
			%\label{eq: green}
			%\mathcal{G}_{\la}(\rv; \rvp) = .... 
		%\eeq
	\end{proposition}

\subsection{Interacting anyons}
\label{sec: interacting}

Once the Hamiltonian of two non-interacting anyons has been rigorously defined as a suitable self-adjoint operator belonging to the family $ \hamb $, it is natural to wonder whether the same results apply to a system of two anyons with pairwise interaction. As discussed in the Introduction, this amounts to study the self-adjoint realizations of the operator
\bdm
	\hamv = \ham + V,
\edm
defined, e.g., on the domain of smooth functions with support away from the origin in $ \Le $. The generalization can not however hold true for any $ V $, because the properties of the domain of the Friedrichs extension (if any) strongly depends on the potential $ V $. A trivial case in which Theorem \ref{teo: closed} and Corollary \ref{cor: hamb} do apply is given by small perturbations of $ \hamf $ in the sense of Kato: if $ V $ is Kato-small either in operator or quadratic form sense w.r.t. $ \ham $, self-adjointness of $ \hamvf $, i.e., the Friedrichs extension of $ \hamv $, easily follows. The remaining operators $ \hamvb $ in the family are then obtained by perturbing the corresponding $ \hamb $.

As a starting point of our discussion, we observe that if $ V $ is real and bounded from below (as we are going to assume in next Ass. \ref{ass: V}), it certainly admits self-adjoint extensions, because it commutes with the complex conjugation. The number of parameters to label such extensions is however unknown.

%In this Sect. we want to extend the results proven for free anyons in presence of an interaction potential, which is not small w.r.t. to $ \ham $.  Before proceeding further, however, we have to 
Let us specify the assumptions we make on the interaction:

	\begin{assumption}[Interaction potential $ V $]
		\label{ass: V}
		\mbox{}	\\
		Let $ V = V(r) $ be a real radial function and let $ V_{\pm} $ denote the positive and negative parts of $ V $, respectively, i.e., $ V = V_+ - V_- $. Then, we assume that
		\begin{itemize}
			\item $ V_+ \in L^2_{\mathrm{loc}}(\R^+) \cap L^{\infty}([0,\eps)) $;
			\item $ V_- \in L^{\infty}(\R^+) $.
		\end{itemize}
	\end{assumption}
	
The above assumptions are certainly not optimal but at least imply that $ H_V : = - \Delta + V  $ is a semi-bounded self-adjoint operator (see, e.g., \cite{LL}) with domain $ \dom(H_V) \subset H^2(\R^2) $. In fact, the same thing holds for $ \hamv $: let $ \formvf $ be the quadratic form
\beq
	\label{eq: formvf}
	\formvf[\psi] : = \formf[\psi] + \int_{\R^2} \diff \rv \: V(r) \lf| \psi \ri|^2,
\eeq
with domain
\beq
	\label{eq: dom formvf}
	\dom\lf[ \formvf \ri] : = \overline{C^{\infty}_0(\R^2 \setminus \{ 0 \})}^{\lf\| \; \ri\|_{\alpha,V}} \cap \Le,
\eeq
where
%\beq
	$\lf\| \psi \ri\|_{\alpha,V}^2 : = \formvf[\psi] + V_0 \lf\| \psi \ri\|_2^2 $,
%\eeq
and 
\beq
	V_0 : = \sup_{r \in \R^+} V_-(r).
\eeq

%The analogue of Proposition \ref{pro: formf} is the following 

	\begin{proposition}[Friedrichs extension]
		\label{pro: formvf}
		\mbox{}	\\
		Let Ass. \ref{ass: V} hold true. Then, the quadratic form $ \formvf $ is closed and bounded from below on $ \dom[\formvf] $ for any $ \alpha \in [0,1] $. Furthermore, for any $ \alpha \in (0,1) $,
		\beq
			\label{eq: dom formvf zero}
			\dom[\formvf] \subset H^1(\R^2).
		\eeq
		The associated self-adjoint operator $ \hamvf $ acts as $ \hamv $ on the domain
		\beq
			\label{eq: dom hamvf}
			\dom\lf( \hamvf \ri) = \lf\{ \psi \in \dom\lf( \hamf \ri) \: \big| \: V \psi \in L^2 \ri\}.
		\eeq
	\end{proposition}
	
	In the proof of the above result is obviously a key point the assumption that $ V $ is bounded in $ L^{\infty}$  above and below in a neighborhood of the origin as well as the fact that $ V $ is radial. While the latter one is less relevant and might be relaxed, the behavior of $ V $ close to the origin affects the asymptotics of functions in $ \dom(\hamvf) $ there.
	
	Given Prop. \ref{pro: formvf}, we can extend Thm. \ref{teo: closed}. We set for $ \alpha \in (0,1) $ and $ \beta \in \R $
\beq
	\label{eq: formvb}
	\framebox{$ \formvb[\psi] = \formv[\phila] +\la^2 \lf\| \phila \ri\|^2_2 - \la^2 \lf\| \psi \ri\|_2^2 + \lf( \beta + \cal \la^{2\alpha} \ri) \lf| q \ri|^2, $}
\eeq
where $ \psi $ belongs to the domain
\beq
	\label{eq: dom formvb}
	\dom[\formvb] = \lf\{ \psi \in \Le \: \big| \: \psi = \phila + q \gila, \phila \in \dom[\formvf], q \in \C \ri\},
\eeq
and $ G_{\la} $ , $ \la \in \R^+ $, and $ \cal $ have been defined in \eqref{eq: gila} and \eqref{eq: cal}, respectively. Again, the Friedrichs form $ \formf $ is included in the family and recovered for $ \beta = \infty $. 

Note that the decomposition in the domain \eqref{eq: dom formvb} is given in terms of $ \gila $, which is {\it not} a deficiency function for $ \hamv $, but rather for $ \ham $. However, this has no consequences at the level of quadratic forms: although we do not explicitly know $ \ker(\hamv^* + \la^2) $, any function belonging there must have the same singularity at the origin as $ \gila $. Furthermore, the decomposition is well-posed because $ \gila - G_{\la'} \in \dom[\formvf] $, for $ \la \neq \la' $, thanks to Prop. \ref{pro: formvf}.

%We complete this Sect. by showing that the forms $ \formvb $ are closed and bounded from below. Next, we characterize the associated Schr\"{o}dinger operators. Notice that we do not know whether such operators cover all possible self-adjoint extensions of $ \hamv $, since we have no access at its deficiency spaces.

	\begin{theorem}[Closedness and boudedness from below of $ \formvb $]
		\label{teo: closed v}
		\mbox{}	\\
		Let Ass. \ref{ass: V} hold true. Then, for any $ \alpha \in (0,1) $ and any $ \beta \in \R $, the quadratic forms $ \formvb $ is closed and bounded from below on the domain $ \dom[\formvb] $. Furthermore,
		\beq
			\label{eq: lb formvb}
			\frac{\formvb[\psi]}{\lf\| \psi \ri\|_2^2} \geq
			\begin{cases}
				- V_0,		&		\mbox{if } \beta \geq 0;	\\
				- \lf( \disp\frac{|\beta|}{\cal} \ri)^{\frac{1}{\alpha}} - V_0,		&		\mbox{if }	 \beta < 0.
			\end{cases}
		\eeq
	\end{theorem}

	\begin{corollary}[Self-adjoint operators $ \hamvb $]
		\label{cor: hamvb}
		\mbox{}	\\
		Let Ass. \ref{ass: V} hold true. Then, the one-parameter family of self-adjoint operators associated to the forms $ \formvb $, $ \alpha \in (0,1) $ and $ \beta \in \R $, is given by
		\beq
			  \lf( \hamvb + \la^2 \ri) \psi = \lf( \hamv + \la^2 \ri) \phila,	
		\eeq
		\bml{
			  \dom\lf( \hamvb \ri) =  \lf\{ \psi \: \Big| \: \psi_n \in \dom\lf(\hamvf\ri), \forall n \neq 0; \; \psi_0 = \phila + q \gila, \phila \in \dom\lf(\hamvf\ri), \ri. 	\\
			 \lf. q = \frac{1}{\beta + \cal \la^{2\alpha}} \bigg[ \meanlrlr{\gila}{V}{\phila} - \frac{\Gamma(\alpha)}{2^{1-\alpha}}  \lim_{r \to 0^+} \frac{\alpha \phila(r) + r \phila^{\prime}(r)}{r^{\alpha}} \bigg] \ri\},	 \label{eq: dom hamvb}
		}
		where $ \lambda > 0 $ is free to choose provided $ \beta + \cal \la^{2\alpha} \neq 0 $.
		Furthermore, the operators $ \hamvb $ extend $ \hamv $, i.e., $ \lf. \hamv \ri|_{\dom(\hamv)} = \hamv $.
	\end{corollary}
	
	\begin{remark}[Definition of $ q $]
		\label{rem: bc V}
		\mbox{}	\\
		It is interesting to observe that the boundary condition in \eqref{eq: dom hamb} has been replaced in \eqref{eq: dom hamvb} by a non-local condition defining the complex parameter $ q $. Indeed, the scalar product on the r.h.s. does not depend only on the behavior of $ \phila $ at the origin, but rather on the function $ \phila $ itself. Notice however that, since $ \phila \in \dom(\hamvf) $, $ V \phila \in L^2(\R^2) $ and therefore the quantity is well defined. The reason behind the different definition of $ q $ is that $ \gila $ does not belong to $ \ker(\hamv^* + \la^2) $.
	\end{remark}

\medskip
\noindent
{\bf Acknowledgements.} 
MC thanks \textsc{D. Lundholm} (KTH, Stockholm) and \textsc{A. Teta} (``Sapienza'' University of Rome) for helpful discussions about the topic of the work.

\section{Proofs}
\label{sec:proofs}

	%In the first part (Sect. \ref{sec: free proof}) of this Sect. we focus on the quadratic forms \eqref{eq: formb} for non-interacting anyons and prove the main result, i.e., Theorem \ref{teo: closed}. Next, we study the associated self-adjoint operators to prove Corollary \ref{cor: hamb} and investigate their spectral properties. Finally, in Sect. \ref{sec: interacting proof}, we disuss interacting anyonic particles.
	
	\subsection{Non-interacting anyons}
	\label{sec: free proof}
	The first result we prove is Prop. \ref{pro: formf}:

	\begin{proof}[Proof of Proposition \ref{pro: formf}]
		Closedness is a trivial consequence of the definition \eqref{eq: dom formf}. Moreover it is not difficult to see that $ \form[\phi] \geq C \lf\| \phi \ri\|_{H^1(\R^2)}^2 $, which implies \eqref{eq: dom formf zero}.
		
		In order to derive \eqref{eq: dom hamf}, we use the definition of the self-adjoint operator $ \hamf $ associated to the quadratic form $ \formf $: the domain of $ \hamf $ is given by
		\beq
			\label{eq: dom op}
			\dom\lf(\hamf\ri) = \lf\{ v \in \dom\lf[\formf\ri] \: \big| \: \exists w \in \Le, \formf\lf[u,v\ri] = \braket{u}{w}, \forall u \in \dom\lf[\formf\ri] \ri\},
		\eeq
		where $ \formf[u,v] $ is the sesquilinear form constructed from $ \formf $, e.g., by polarization (recall \eqref{eq: nablaa})
		\beq
			\label{eq: sesq formf}
			\formf[u,v] = \int_{\R^2} \diff \rv \: \nablaa u^* \cdot \nablaa v.
		\eeq
		Moreover, in the notation of the above definition, 
		\beq
			w = : \hamf v.
		\eeq 
		Integrating by parts the first term in \eqref{eq: sesq formf} and using that $ u $ and $ v $ both vanish in a neighborhood of $ 0 $, the identification of the action of $ \hamf $ with the one of $ H_{\alpha} $ is immediate. The vanishing of any $ \psi \in \dom[\formf] $ at $ 0 $ can be seen from the decomposition \eqref{eq: formf alternative}: for any $ \alpha \in (0,1) $ and for any $ k \in \Z  $, it must be
		\beq
			\int_{\R^2} \diff \rv \: \lf| \psi_{2k}^{\prime} \ri|^2 < + \infty,	\qquad		\int_{\R^2} \diff \rv \: r^{-2} \lf| \psi_{2k} \ri|^2 < + \infty,
		\eeq
		since $ 2k - \alpha \neq 0 $, $ \forall k \in \Z $. Hence, both conditions are met at the same time only if $ \psi_{2k} \to 0  $, as $ r \to 0 $, for any $ k \in \Z $. Notice also that the second condition for $ k \neq 0 $ is in fact implicitly contained in the hypothesis $ \psi \in H^1(\R^2) $.
		
		 Any $ \psi \in \dom(\hamf) $ must then be such that $ \psi \in \dom[\formf] $ and $ H_{\alpha} \psi \in L^2(\R^2) $. If we decompose $ \psi $ in cylindrical harmonics as in \eqref{eq: fourier}, %, i.e.,
		%\bdm
			%\psi(\rv) = \sum_{k \in \Z} \psi_{2k}(r) e^{2ik  \vartheta},
		%\edm
		this is equivalent to the request $ \psi_{2k} \in H^2(\R^2) $, for any $ k \neq 0 $: for any function $ \psi  \in H^2(\R^2) \cap \HH_{2k} $, $ k \neq 0 $, we have
		\bdm
			\lf( - \Delta \psi \ri)(\rv) = \lf[ - \psi^{\prime\prime}_{2k}(r) - \frac{1}{r} \psi_{2k}^{\prime}(r)  + \frac{4 k^2}{r^2} \psi_{2k}(r) \ri] e^{2ik\vartheta},
		\edm
		and obviously any such $ \psi $ must vanish faster than $ r $ at the origin, in such a way that the last term above, i.e., $ r^{-2} \psi $ is square integrable. The AB potential squared, which is also proportional to $ r^{-2} $, is then automatically bounded on such functions.
		
		 On the contrary, in $ \HH_0 $, we may have that $ \psi_0 \in \dom(\hamf) $ but $ \psi_0 \notin  H^2(\R^2) $. Indeed, $ \psi_0 $ belongs to $ \dom(\hamf) $, if either $ - \Delta \psi_0 \in L^2(\R^2) $ and $ r^{-2} \psi_0 \in L^2(\R^2) $, which implies that 
		\beq
		 	\psi_0 \in H^2(\R^2),	\quad	\mbox{and}	\quad	 \psi_0 = o(r), \mbox{ as }  r \to 0^+,
		\eeq
		or the single terms are not in $ L^2 $ but the sum does. Since, however, $ \psi_0 \in \dom[\formf] $ and therefore it must vanish at the origin, we can assume that $ \psi_0(\rv) \sim  r^{\nu} (1 + o(1)) $, as $ r \to 0^+ $, for some $ 0 < \nu < 1  $, which implies that $ \psi_0 \notin H^2(\R^2) $, and
		\beq
			\label{eq: cond dom hamf}
			\lf( -\nu^2 + \alpha^2 \ri) r^{\nu - 2}	\in L^2(\R^2),
		\eeq
		which requires
		\beq
			\label{eq: hamf n asympt}
			\nu  =  \alpha.
		\eeq
		Notice that the vanishing of the factor in \eqref{eq: cond dom hamf} is met also if $ \nu = -  \alpha $, but this is not acceptable because $ \psi \in \dom[\formf] $ and therefore it must vanish at $ 0 $.
	\end{proof}
		
We can now focus on the quadratic form $ \formb $. As shown in Sect. \ref{sec: main}, the definition is well-posed and $ \dom[\formb] $ is actually independent of the choice of the parameter $ \la > 0 $. We are now going to prove the main result of the paper:

	\begin{proof}[Proof of Theorem \ref{teo: closed}]
		The proof is divided in two steps: first, we show that the quadratic form $ \formb $ is bounded from below for any $ \alpha \in (0,1) $ and $ \beta \in \R $; next, we prove closedness. In both steps, the freedom in the choice of $ \la > 0 $ will be key.
		
		From \eqref{eq: formb alternative} one immediately gets the inequality
		\beq
			\formb[\psi] \geq - \la^2 \lf\| \psi \ri\|^2_2 + \lf( \beta + \cal \la^{2\alpha} \ri) |q|^2,
		\eeq
		where the coefficient $ \cal > 0 $ is given by \eqref{eq: cal}. Now, if $ \beta \geq 0 $, there is nothing to prove and one actually obtains $ \formb[\psi] \geq 0 $, by taking $ \la $ arbitrarily small. If on the opposite $ \beta < 0 $, we pick 
		\beq
			\label{eq: la lb}
			\la = \lf( \frac{|\beta|}{\cal} \ri)^{\frac{1}{2 \alpha}},
		\eeq
		which leads to \eqref{eq: lb formb}.
		
		To prove closedness, we first pick $ \la > \bar{\la} $, where
		\beq
			\label{eq: barla}
			\bar{\la} =
			\begin{cases}
				0,	&	\mbox{if } \beta \geq 0,	\\
				\lf( \frac{|\beta|}{\cal} \ri)^{\frac{1}{2 \alpha}},	&	\mbox{otherwise},
			\end{cases}
		\eeq
		and investigate the form $ \formt[\psi] := \formb[\psi] + \la^2 \lf\| \psi \ri\|_2^2 $. Of course, closedness of $ \formt $ is totally equivalent to closedness of $ \formb $. Let us then take a sequence $ \lf\{ \psi_n \ri\}_{n \in \N} $ of wave functions belonging to $  \dom[\formb] $ and such that
		\beq
			\label{eq: cauchy}
			\lim_{n, m \to \infty} \formt\lf[\psi_n - \psi_m\ri] = 0,		\qquad		\lim_{n, m \to \infty} \lf\| \psi_n - \psi_m\ri\|_2^2 = 0.
		\eeq
		Let $ \psi \in \Le $ be the limit of the sequence, we want to prove that $ \psi \in \dom[\formb] $. 
		
		Since $ \psi_n \in \dom[\formb] $ for any $ n $, we can decompose
		\beq
			\psi_n = \phi_n + q_n \gila,
		\eeq
		where $ \phi_n \in \dom[\formf] $ and $ q_n \in \C $. Thanks to the choice of $ \la $, we have
		\bdm
			\formt\lf[\psi_n - \psi_m\ri] \geq \formf\lf[ \phi_n - \phi_m \ri] + \la^2 \lf\| \phi_n - \phi_m \ri\|_2^2 + \lf( \beta + \cal \la^{2\alpha} \ri) |q_n - q_m|^2,
		\edm
		which implies, by positivity of $ \formf $, that $ \formf\lf[\phi_n - \phi_m\ri] \to 0 $ and		%\qquad		
		$  \lf\| \phi_n - \phi_m\ri\|_2^2 \to 0 $.
		%\eeq
		Hence, by closedness of $ \formf $, $ \phi_n \to \phi \in \dom[\formf] $. Moreover, we also have $  \lf| q_n - q_m \ri| \to 0 $ and thus $ \lf\{ q_n \ri\}_{n \in \N} $ is a Cauchy sequence in $ \C $, which implies that $ q_n \to q \in \C $. In conclusion,
		\beq
			\psi_n \xrightarrow[n \to \infty]{} \phi + q \gila \in \dom[\formb],
		\eeq
		and the result is proven.
	\end{proof}

	Once it is known that the forms $ \formb $ are closed on their domain $ \dom[\formb] $, it is natural to find out what is the associated self-adjoint operator. This leads to Cor. \ref{cor: hamb}.
	
	\begin{proof}[Proof of Corollary \ref{cor: hamb}]
		In order to figure out what is the operator $ \hamb $ associated to $ \formb $, we apply the definition \eqref{eq: dom op}, where the sesquilinear form associated to $ \formb $ (constructed, e.g., by polarization) reads
		\bml{
			\label{eq: sesquilinear}
			\formb[u,v] = \form[\phi_u,\phi_v] - q_u^* \braket{\gila}{\phi_v} - q_v \braket{\phi_u}{\gila} + \beta q_u^* q_v	\\
			= \form[\phi_u,\phi_v] + \la^2 \braket{\phi_u}{\phi_v} - \la^2 \braket{u}{v}  + \lf( \beta + \cal \la^{2\alpha} \ri)  q_u^* q_v,
		}
		w.r.t. the decompositions $ u = \phi_u + q_u \gila $, $ v = \phi_v + q_v \gila $, valid for any $u, v \in \dom[ \formb] $.
		
		Let us first assume that $ u \in \dom[\formb] $ varies in the dense subset given by $ \dom[\formf] $. If this is the case, then $ u = \phi_u $, i.e., $ q_u = 0 $, and the form simplifies:
		\beq
			\formb[u,v] = \formf[\phi_u,\phi_v] + \la^2 \braket{\phi_u}{\phi_v} - \la^2 \braket{\phi_u}{v} = \braket{\phi_u}{w},
		\eeq
		where $ w \in \Le $ and the identity has to be satisfied at least for any $ u \in  \dom[\formf] $. Since $ \phi_v, v \in \Le $ by assumption, the above identity implies that there exists $ w' \in \Le $ such that
		 \bdm
		 	\formf[\phi_u,\phi_v] = \braket{\phi_u}{w'},
		\edm
		for any $ \phi_u \in \dom[\formf] $. This is the definition of the domain of $ \hamf $ (compare with the proof of Prop. \ref{pro: formf}) and therefore we conclude that 
		\beq
			\phi_v \in \dom\lf( \hamf \ri).
		\eeq
		Moreover, 
		\beq
			\label{eq: action hamb}
			\hamb v : = w = \hamf \phi_v + \la^2 \phi_v - \la^2 v,
		\eeq
		which provides the action of the operator. 
		
		Next, we pick a generic $ u \in \dom[\formb] $, i.e., with $ q_u \neq 0 $, and get
		\bml{
			\formf[\phi_u,\phi_v] + \la^2 \braket{\phi_u}{\phi_v} - \la^2 \braket{u}{v} + \lf( \beta + \cal \la^{2\alpha} \ri)  q_u^* q_v = \braket{u}{w} 	\\
			%= \braketr{u}{\hamf \phi_v + \la^2 \phi_v - \la^2 v} \\
			=  \braketr{\phi_u}{\hamf \phi_v} + \la^2 \braketr{\phi_u}{\phi_v} - \la^2 \braket{u}{v} + q^*_{u} \braketr{\gila}{\lf( \hamf + \la^2 \ri) \phi_v},
		}
		where we have made use of \eqref{eq: action hamb}. Now, the first three terms exactly cancel with the first three on the l.h.s., so that we get the equation
		\beq
			 \lf( \beta + \cal \la^{2\alpha} \ri)  q_u^* q_v = q_u^* \braketr{\gila}{\lf( \hamf + \la^2 \ri) \phi_{v,0}},
		\eeq
		which has to be satisfied for any $ q_u \in \C $. Notice that only the $0-$th Fourier component $ \phi_{v,0} $ appears, because $ \hamf $ commutes with rotations and $ \gila \in \HH_0 $. Hence, if $ \beta + \cal \la^{2\alpha -2} \neq 0 $, which we can always assume,
		\beq
			\label{eq: bc proof}
			\lf( \beta + \cal \la^{2\alpha} \ri) q_v = \braketr{\gila}{\lf( \hamf + \la^2 \ri) \phi_{v,0}}.
		\eeq
		On the other hand,
		\beq
			\braketr{\gila}{\lf( \hamf + \la^2 \ri) \phi_{v,0}} = \int_{0}^{\infty} \diff r \, r \: \gila(r) \bigg[ \bigg( - \Delta + \frac{\alpha^2}{r^2} + \la^2 \bigg) \phi_{v,0} \bigg](r),
		\eeq
		and a direct integration by parts yields 
		\bml{
			\label{eq: int by parts}
			\int_{0}^{\infty} \diff r \, r \: \gila \lf( - \Delta \phi_{v,0} \ri) = - \int_{0}^{\infty} \diff r \: \gila \; \partial_r \lf( r \phi^{\prime}_{v,0}  \ri) \\
			%= - \lf| r \gila(r) \phi^{\prime}_{v,0}(r) \ri|^{+\infty}_{0^+} + \int_{0}^{\infty} \diff r \, r \: \gila^{\prime}(r) \phi^{\prime}_{v,0}(r)\\
			= \lf| r \gila^{\prime}(r) \phi_{v,0}(r) \ri|^{+\infty}_{0^+} - \lf| r \gila(r) \phi^{\prime}_{v,0}(r) \ri|^{+\infty}_{0^+} + \int_{0}^{\infty} \diff r \, r \: \lf( - \Delta \gila \ri) \phi_{v,0}.
		}
		Combining \eqref{eq: defect} with the vanishing of any $ \phi_{v,0} \in \dom(\hamf) $ at the origin by Prop. \ref{pro: formf}, which guarantees that $ \gila \in \dom(\hamf^*) $, we deduce
		\beq
			%\label{eq: defect}
			\lf( \hamf^* + \la^2 \ri) \gila = 0.
		\eeq
		Applying it to \eqref{eq: bc proof} via \eqref{eq: int by parts}, we conclude that
		\beq
			\lf( \beta + \cal \la^{2\alpha} \ri) q_v =  - \lf| r \gila(r) \phi^{\prime}_{v,0}(r) \ri|^{+\infty}_{0^+} + \lf| r \gila^{\prime}(r) \phi_{v,0}(r) \ri|^{+\infty}_{0^+},
		\eeq
		which then yields the boundary condition in \eqref{eq: dom hamb}, if one exploits the asymptotics of $ \gila $ for $ r \to 0^+ $ \eqref{eq: gila asympt}. The condition $ \beta + \cal \la^{2\alpha} \neq 0 $ is needed to ensure that the above condition can be inverted to obtain $ q_v $ from a given $ \phi_{v,0} $.		
	\end{proof}
	
	%\begin{remark}[Bosons]
		%\label{rem: bosons proof}
		%\mbox{}	\\
		%The proof argument applies to the case of bosons $ \alpha = 0 $ too. The only relevant difference is the asymptotics of $ K_0(\la r) $, which reads $ K_0(\la r) = - \log r + \log \la/2 + \gamma_{\mathrm{E}} + \OO(r^2) $, $ \gamma_{\mathrm{E}} $ being the Euler constant. If such a behavior is properly taken into account, the boundary condition \eqref{eq: bc bosons} is recovered.
	%\end{remark}
	
	\begin{proof}[Proof of Proposition \ref{pro: spectrum}]
		Since the deficiency indices of $ \ham $ are both equal to $ 1 $, standard results in operator theory imply that all the self-adjoint extensions of $ \ham $ have the same essential spectrum of $ \ham $ and the discrete spectrum consists in at most one eigenvalue. 
		
		It remains then just to show that, if $ \beta < 0 $, the eigenvalue (ground state) is given by \eqref{eq: eigenvalue}. To this purpose we observe that $ G_{\bar{\lambda}} $ is an eigenfunction of $ \hamb $, if $ \bar\lambda $ is given by \eqref{eq: la lb}: by \eqref{eq: hamb action} it suffices to show that $ G_{\bar{\la}} \in \dom(\hamb) $, i.e., we can take $ \phi_{\bar{\la}} = 0 $ but $ q \neq 0 $. This is however perfectly consistent with \eqref{eq: dom hamb}, since the r.h.s. of the boundary condition vanishes when $ \la = \bar{\la} $.
	\end{proof}
	
	\subsection{Interacting anyons}
	\label{sec: interacting proof}
	The key result is Prop. \ref{pro: formvf}, since all the other statements follow as in the non-interacting case.
	
	\begin{proof}[Proof of Proposition \ref{pro: formvf}] Closedness and positivity are obvious, so let us consider the characterization of the domain. Since $ V \in L^{\infty}([0,\eps)) $ by Ass. \ref{ass: V}, the asymptotic behavior at the origin of functions in $ \dom(\formvf) $ as well as in $ \dom(\hamvf) $ is not affected by the presence of $ V $. The result then easily follows.
	\end{proof}

	We can now prove the main result in the interacting case:
	
	\begin{proof}[Proof of Theorem \ref{teo: closed v}]
		We first pick $ \la > \bar{\la} $, where $ \bar{\la} $ is the lower bound stated in \eqref{eq: lb formvb} and consider the form $ \formt[\psi] := \formb[\psi] + \la^2 \lf\| \psi \ri\|_2^2 $, which is then positive by assumption. We then take a sequence $ \lf\{ \psi_n \ri\}_{n \in \N} \subset  \dom[\formb] $, such that \eqref{eq: cauchy} holds true. After decomposing $ \psi_n = \phi_n + q_n \gila $, we deduce as in the proof of Thm. \ref{teo: closed} that $ \lf\{ \phi_n \ri\}_{n \in \N} $ and $ \lf\{ q_n \ri\}_{n \in \N}$ are Cauchy sequences w.r.t. the norm $ \lf\| \: \cdot \: \ri\|_{\alpha,V} $ and in $ \C $, respectively. Hence, the limit $ \psi $ decomposes as well as $ \phi_{\la} + q \gila $, with $ \phi_{\la} \in \dom[\formvf] $ and $ q \in \C $, i.e., $ \psi \in \dom[\formvb] $.
	\end{proof}
	
	The proof of Cor. \ref{cor: hamvb} is identical to the proof of Cor. \ref{cor: hamb}. We omit it for the sake of brevity.

\end{document}